\documentclass{article}
\usepackage[T1]{fontenc}
\usepackage[utf8]{inputenc}
\usepackage{lmodern}
\usepackage{xspace}                                     
\usepackage{fullpage}
\usepackage{amsfonts,amsmath,amssymb, amsthm}
\usepackage{dsfont} 
\usepackage{algorithmicx,algpseudocode,algorithm}
\usepackage[usenames,dvipsnames,table]{xcolor}
\usepackage{tikz}

\usepackage{multirow} 
\usepackage{chngpage} 
\usepackage{mfirstuc}
\usepackage[normalem]{ulem}

\usepackage[shortlabels]{enumitem}
 \setitemize{noitemsep,topsep=3pt,parsep=1pt,partopsep=1pt} 
 \setenumerate{itemsep=1pt,topsep=1pt,parsep=1pt,partopsep=1pt}
 \setdescription{itemsep=1pt}

\def\fulldetails{1}
\def\icalpstyle{0}
\input{preamble}

\title{Testing probability distributions underlying aggregated data}
\author{
  \cnote{Cl\'{e}ment Canonne}\thanks{Columbia University. Email: \email{ccanonne@cs.columbia.edu}}
  \and \rnote{Ronitt Rubinfeld}\thanks{CSAIL, MIT and the Blavatnik School of Computer Science, Tel Aviv University. Email: \email{ronitt@csail.mit.edu}. Research supported by NSF grants CCF-1217423 and CCF-1065125.}
}
\date{February 15, 2014}

\newcommand{\pdfsamp}{dual\xspace}
\newcommand{\cdfsamp}{cumulative dual\xspace}
\newcommand{\Pdfsamp}{Dual\xspace}
\newcommand{\Cdfsamp}{Cumulative Dual\xspace}
\newcommand{\D}{\ensuremath{D}}

\begin{document}

\maketitle

\begin{abstract}
In this paper, we analyze and study a hybrid model for testing and learning probability distributions. Here, in addition to samples, the testing algorithm is provided with one of two different types of oracles to the unknown distribution $\D$ over $[n]$. More precisely, we define both the \emph{\pdfsamp} and \emph{\cdfsamp access models}, in which the algorithm $A$ can both sample from $\D$ and respectively, for any $i\in[n]$,
\begin{itemize}
  \item query the probability mass $D(i)$ \emph{(query access)}; or
  \item get the total mass of $\{1,\dots,i\}$, i.e. $\sum_{j=1}^i D(j)$ \emph{(cumulative access)}
\end{itemize}
These two models, by generalizing the previously studied sampling and query oracle models, allow us to bypass the strong lower bounds established for a number of problems in these settings, while capturing several interesting aspects of these problems -- and providing new insight on the limitations of the models. Finally, we show that while the testing algorithms can be in most cases strictly more efficient, some tasks remain hard even with this additional power.
\end{abstract}

\section{Introduction}
  \subsection{Background}
  
        Given data sampled from a population or an experiment, understanding the distribution from which it has been drawn is a fundamental problem in statistics, and one which has been extensively studied for decades. However, it is only rather recently that these questions have been considered when the distribution is over a \emph{very large} domain (see for instance \cite{BFRSW00,GRexp:00,Ma:81:Physics}). In this case, the usual techniques in statistics and learning theory become impractical, motivating the search for better algorithms, in particular by relaxing the goals so that learning is not required. This is useful in many real-world applications where only a particular aspect of the distribution is investigated, such as estimating the entropy or the distance between two distributions. In these examples, as well as many others, one \emph{can} achieve sublinear sample complexity. However, strong lower bounds show that the complexity of these tasks is still daunting, as it has polynomial, and often nearly linear, dependence on the size of the support of the distribution. To address this difficulty, new lines of research have emerged.  One approach is to obtain more efficient algorithms for special classes of distributions. For instance, improved algorithms whose sample complexity is polylogarithmic in the domain size can be achieved by requiring it to satisfy specific smoothness assumptions, or to be of a convenient shape (monotone, unimodal, or a ``$k$-histogram'' \cite{BKR:04,ILR12,DDSVV:13}).
A different approach applies to general distributions, but gives the algorithm more power in form of more flexible access to the distribution: as in many applications the data has already been collected and aggregated, it may be reasonable to assume that the testing algorithm can make other limited queries to the probability density function. For example, the algorithm may be provided with query access to the probability density function of the distribution
\cite{RubinfeldServedio:09}, or samples from conditional distributions induced by the original distribution \cite{CFGM:13,CRS12,CRS14}.
    
  \subsection{Our model: \pdfsamp and \cdfsamp oracles}  
  In this work, we consider the power of two natural oracles.
The first is a \emph{dual oracle}, which combines the standard model for distributions and the familiar one commonly assumed for testing Boolean and real-valued functions. In more detail, the testing algorithm is granted access to the unknown distribution $\D$ through two independent oracles, 
one providing samples of the distribution, while the other, on query $i$ in the domain of the distribution, provides the value of the probability
density function at $i$. \footnotemark{}

\footnotetext{Note that in both definitions, one can decide to disregard the corresponding evaluation oracle, which in effect amounts to falling back to the standard sampling model; moreover, for our domain $[n]$, any $\EVAL_\D$ query can be simulated by (at most) two queries to a $\CDFEVAL_\D$ oracle -- in other terms, the \cdfsamp model is at least as powerful as the \pdfsamp one.}
    
    \begin{definition}[\Pdfsamp access model]
    Let $\D$ be a fixed distribution over $[n]\newer{=\{1,\dots,n\}}$. A \emph{\pdfsamp oracle for $\D$} is a pair of oracles $(\SAMP_\D, \EVAL_\D)$ defined as follows: when queried, the \emph{sampling} oracle $\SAMP_\D$ returns an element $i\in[n]$, where the probability that $i$ is returned is $\D(i)$ independently of all previous calls to any oracle; while the \emph{evaluation} oracle $\EVAL_\D$ takes as input a query element $j\in[n]$, and returns the probability weight $\D(j)$ that the distribution puts on $j$.
    \end{definition}
    It is worth noting that this type of \pdfsamp access to a distribution has been considered (under the name \emph{combined oracle}) in \cite{BDK+:05} and \cite{GMV:05}, where they address the task of estimating (multiplicatively) the entropy of the distribution, or the $f$-divergence between two of them (see Sect.~\ref{sec:supportsize:entropy} for a discussion of their results).\medskip
    
    The second oracle that we consider provides samples of the distribution as well as queries to the \emph{cumulative distribution function} (cdf) at any point in the domain\footnotemark{}. 
    \footnotetext{We observe that such a cumulative evaluation oracle $\CDFEVAL$ appears in \cite{BKR:04} (Sect.~8).}
    \begin{definition}[\Cdfsamp access model]
    Let $\D$ be a fixed distribution over $[n]$. A \emph{\cdfsamp oracle for $\D$} is a pair of oracles $(\SAMP_\D, \CDFEVAL_\D)$ defined as follows: the \emph{sampling} oracle $\SAMP_\D$ behaves as before, while the \emph{evaluation} oracle $\CDFEVAL_\D$ takes as input a query element $j\in[n]$, and returns the probability weight that the distribution puts on $[j]$, that is
    $\D([j])=\sum_{i=1}^j D(i)$ .
    \end{definition}    
    
  \subsection{Motivation and discussion}\label{sec:motivation}
  As a first motivation to this hybrid model, consider the following scenario: There is a huge and freely available dataset, which a computationally-limited party -- call it Arthur -- needs to process. Albeit all the data is public and Arthur can view any element of his choosing, extracting further information from the dataset (such as the number of occurrences of a particular element) takes too much time. However, a third-party, Merlin, has already spent resources in preprocessing this dataset and is willing to disclose such information~--~yet at a price. This leaves Arthur with the following question: \emph{how can he get his work done as quickly as possible, paying as little as possible?} This type of question is captured by our new model, and can be analyzed in this framework. For instance, if the samples are stored in sorted order, implementing either of our oracles becomes possible with only a logarithmic overhead per query. \newer{It is worth noting that Google has published their $N$-gram models, which describe their distribution model on 5-word sequences in the English language. In addition, they have made available the texts on which their model was constructed. Thus, samples of the distribution in addition to query access to probabilities of specific domain elements may be extracted from the Google model.}
  
  A second and entirely theoretical motivation for studying distribution testing in these two dual oracle settings arises from attempting to understand the limitations and underlying difficulties of the standard sampling model. Indeed, by circumventing the lower bound, one may get a better grasp on the core issues whence the hardness stemmed in the first place.
  
  Another motivation arises from data privacy, when a curator administers a database of highly sensitive records (e.g, healthcare information, or financial records). 
  Differential privacy \cite{DN:03,DN:04,Dwork:08} studies mechanisms which allow the curator to release relevant information about its database without without jeopardizing the privacy of the individual records.
  In particular, mechanisms have been considered that enable the curator to \emph{release} a sanitized approximation $\tilde{\D}$ of its database $\D$, which ``behaves'' essentially the same for all queries of a certain type -- such as \emph{counting} or \emph{interval queries}\footnotemark{} \cite{BLR:13}.
   \newer{Specifically, if the user needs to test a property of a database, it is sufficient to test whether the sanitized database has the property, using now both samples and interval (i.e., $\CDFEVAL$) or counting ($\EVAL$) queries. As long as the tester has some tolerance (in that it accepts databases that are close to having the property), it is then possible to decide whether the true database itself is close to having the property of interest.}
   \footnotetext{A counting query is of the form ``how many records in the database satisfy predicate $\chi$?'' -- or, equivalently, ``what is the probability that a random record drawn from the database satisfies $\chi$?''.}
     
  Finally, a further motivation is the tight connection between the \pdfsamp access model and the \emph{data-stream model}, as shown by Guha et al. (\cite{GMV:05}, Theorem 25): more precisely, they prove that any (multiplicative) approximation algorithm for a large class of functions of the distribution (functions that are invariant by relabeling of any two elements of the support) in the \pdfsamp access model yields a space-efficient, $\bigO{1}$-pass approximation algorithm for the same function in the data-stream model.

  \subsection{Our results and techniques}
  We focus here on four fundamental and pervasive problems in distribution testing, which are testing \emph{uniformity}, \emph{identity} to a known distribution $\D^\ast$, \emph{closeness} between two (unknown) distributions $\D_1$, $\D_2$, and finally \emph{entropy and support size}. As usual in the distribution testing literature, the notion of distance we use is the \emph{total variation distance} (or statistical distance), which is essentially the $\ell_1$ distance between the probability distributions (see Sect.~\ref{sec:preliminaries} for the formal definition). \newer{Testing closeness is thus the problem of deciding if two distributions are equal or far from each other in total variation distance; while tolerant testing aims at deciding whether they are sufficiently close versus far from each other.}
    
  As shown in Table~\ref{table:results}, which summarizes our results and compares them to the corresponding bounds for the standard sampling-only ($\SAMP$), evaluation-only ($\EVAL$) and conditional sampling ($\COND$) models, we indeed manage to bypass the aforementioned limitations of the sampling model, and give (often tight) algorithms with sample complexity either constant (with relation to $n$) or logarithmic, where a polynomial dependence was required in the standard setting.
  
  Our main finding overall is that \emph{both \pdfsamp models allow testing algorithms to significantly outperform both $\SAMP$ and $\COND$ algorithms}, either with relation to the dependence on $n$ or, for the latter, in $1/\eps$; further, these testing algorithms are \emph{significantly simpler}, both conceptually and in their analysis, and can often be made robust to some multiplicative noise in the evaluation oracle. Another key observation is that this new flexibility 
not only allows us to tell whether two distributions are close or far, but also to efficiently estimate their distance\footnote{For details on the equivalence between tolerant testing and distance estimation, the reader is referred to \cite{PRR:2006}.}.
  
  \newer{In more detail, we show that for the problem of testing equivalence between distributions, both our models allow to get rid of any dependence on $n$, with a (tight) sample complexity of $\bigTheta{1/\eps}$. The upper bound is achieved by adapting an \EVAL-only algorithm of \cite{RubinfeldServedio:09} (for identity testing) to our setting, while the lower bound is obtained by designing a far-from-uniform instance which ``defeats'' simultaneously both oracles of our models. Turning to tolerant testing of equivalence, we describe algorithms whose sample complexity is again independent of $n$, in sharp contrast with the $n^{1-o(1)}$ lower bound of the standard sampling model. Moreover, we are able to show that, at least in the \Pdfsamp access model, our quadratic dependence on \eps is optimal. The same notable improvements apply to the query complexity of estimating the support size of the distribution, which becomes constant (with relation to $n$) in both of our access models -- versus quasilinear if one only allows sampling.
  
  As for the task of (additively) estimating the entropy of an arbitrary distribution, we give an algorithm whose sample complexity is only polylogarithmic in $n$, and show that this is tight in the \Pdfsamp access model, \newer{up to the exponent of the logarithm}. Once more, this is to be compared to the $n^{1-o(1)}$ lower bound for sampling.
  }
  
  \begin{table}[h]\scriptsize\renewcommand{\arraystretch}{1.75}\centering
\ifnum\icalpstyle=1
  \begin{adjustwidth}{-.35in}{-.35in}\centering
\else
  \begin{adjustwidth}{-0.5in}{0.4in}\centering
\fi
    \begin{tabular}{||l||c||c||c||||c||c||}
    \hline\hline
    Problem            & \SAMP           & \COND  \cite{CRS12,CRS14}            & \EVAL          & \Pdfsamp          & \Cdfsamp        \\ \hline\hline
    Testing uniformity  & $\bigTheta{\frac{\sqrt{n}}{\eps^2}}$ \cite{GRexp:00,BFRSW:10,Paninski:08}  & $\tildeO{\frac{1}{\eps^{2}}}$, $\bigOmega{\frac{1}{\eps^{2}}}$ & \multirow{2}{*}{$\bigO{\frac{1}{\eps}}$ \cite{RubinfeldServedio:09}, {$\bigOmega{\frac{1}{\eps}}^\ast$} }  & \multirow{3}{*}{$\bigTheta{\frac{1}{\eps}}$ $(\dagger)$} & \multirow{3}{*}{$\bigTheta{\frac{1}{\eps}}$ $(\dagger)$}\\ \cline{2-3}
    Testing $\equiv D^\ast$ & $\tildeTheta{\frac{\sqrt{n}}{\eps^2}}$ \cite{BFFKRW:01,Paninski:08} & $\tildeO{\frac{1}{\eps^{4}}}$ & & & \\ \cline{2-4}
    Testing $D_1 \equiv D_2$ & $\bigTheta{ (\max\left(\frac{N^{2/3}}{\eps^{4/3}}, \frac{\sqrt{N}}{\eps^{2}} \right) }$ \cite{BFRSW:10,Valiant:11,CDVV:13} & $\tildeO{\frac{\log^5 n}{\eps^{4}}}$ & {$\bigOmega{\frac{1}{\eps}}^\ast$} & & \\ \hline\hline
    Tolerant uniformity & \parbox{40mm}{\centering $\bigO{\frac{1}{(\eps_2-\eps_1)^2}\frac{n}{\log n}}$ \cite{ValiantValiant:11,ValiantValiant:10ub} $\bigOmega{\frac{n}{\log n}}$ \cite{ValiantValiant:11,ValiantValiant:10lb} \strut} & $\tildeO{\frac{1}{(\eps_2-\eps_1)^{20}}}$       & \multirow{3}{*}{{$\bigOmega{\frac{1}{(\eps_2-\eps_1)^{2}}}^\ast$}} & \multirow{3}{*}{$\bigTheta{\frac{1}{(\eps_2-\eps_1)^{2}}}$ $(\dagger)$} & \multirow{3}{*}{$\bigO{\frac{1}{(\eps_2-\eps_1)^{2}}}$ $(\dagger)$} \\ \cline{2-3}
    Tolerant $D^\ast$  & \multirow{2}{*}{$\bigOmega{\frac{n}{\log n}}$ \cite{ValiantValiant:11,ValiantValiant:10lb}}  & {\cellcolor{gray!25}}~ & & &   \\ 
    Tolerant $D_1, D_2$ & & {\cellcolor{gray!25}~}  & &  &   \\ \hline\hline
     \parbox{25mm}{ Estimating entropy to $\pm\Delta$ \strut }      & $\bigTheta{\frac{n}{\log n}}$ \cite{ValiantValiant:11,ValiantValiant:10lb} & \cellcolor{gray!25}~  & \cellcolor{gray!25}~ & $\bigO{\frac{\log^2\frac{n}{\Delta}}{\Delta^{2}}}$  $(\dagger)$, $\bigOmega{\log n}$ & $\bigO{\frac{\log^2\frac{n}{\Delta}}{\Delta^{2}}}$  $(\dagger)$     \\ \hline
    \parbox{25mm}{ Estimating support size to $\pm\eps n$ \strut } & $\bigTheta{\frac{n}{\log n}}$ \cite{ValiantValiant:11,ValiantValiant:10lb} & \cellcolor{gray!25}~  & \cellcolor{gray!25}~\ignore{ \new{$\bigTheta{\frac{1}{\eps^{2}}}$} $(\star)$ }        & $\bigTheta{\frac{1}{\eps^{2}}}$             & $\bigO{\frac{1}{\eps^{2}}}$           \\ \hline\hline
    \end{tabular}
\end{adjustwidth}
    \caption{\label{table:results}\small Summary of results. $(\dagger)$ \ignore{(resp. $(\ddagger)$)} stands for ``robust to multiplicative noise''\ignore{, while $(\star)$ indicates mild restrictions}. The bounds with an asterisk are those which, in spite of being for different models, derive from the results of the last two columns.}
\ifnum\icalpstyle=1
  \vspace{-2\baselineskip}
\fi
\end{table}

While it is not clear, looking at these problems, whether the additional flexibility that the \Cdfsamp access grants over the \Pdfsamp one can \emph{unconditionally} yield strictly more sample-efficient testing algorithms, we do provide a separation between the two models in Sect.~\ref{sec:supportsize:entropy:monotone} by showing an exponential improvement in the query complexity for estimating the entropy of a distribution given the promise that the latter is (close to) monotone. This leads us to suspect that for the task of testing monotonicity (for which we have preliminary results), under a structural assumption on the distribution, or more generally for properties intrinsically related to the underlying total order of the domain, such a speedup holds.  Moreover, we stress out the fact that our $\bigOmega{1/(\eps_2-\eps_1)^2}$ lower bound for tolerant identity testing does not apply to the \Cdfsamp setting.

One of the main techniques we use for algorithms in the \pdfsamp model is a general approach\footnotemark{} for estimating very efficiently any quantity of the form
$\shortexpect_{i\sim \D}\left[ \Phi(i,\D(i)) \right]$, for any \emph{bounded} function $\Phi$. In particular, in the light of our lower bounds, this technique is both an intrinsic and defining feature of the \Pdfsamp model, as it gives essentially tight upper bounds for the problems we consider.

On the other hand, for the task of proving lower bounds, we no longer can take advantage of the systematic characterizations known for the sampling model 
 (see e.g.~\cite{BarYossef:2002}, Sect. 2.4.1). For this reason, we have to rely on reductions from known-to-be-hard problems (such as estimating the bias of a coin), or prove indistinguishability in a \newer{\emph{customized}} fashion.
\footnotetext{We note that a similar method was utilized in~\cite{BDK+:05}, albeit in a less systematic way.}
  
  \subsection{Organization}
  After the relevant definitions and preliminaries in Sect.~\ref{sec:preliminaries}, we \newer{pursue by considering the first three problems of testing equivalence of distributions} in Sect.~\ref{sec:uniformity:equivalence:identity}, where we \newer{describe} our testing upper and lower bounds. \newer{We then turn} to the harder problem of \emph{tolerant} testing. Finally, we tackle in Sect.~\ref{sec:supportsize:entropy} the task of performing entropy and support size estimation, and give for the latter matching upper and lower bounds.

\section{Preliminaries}\label{sec:preliminaries}
We consider discrete probability distributions over the subset of integers $[n]=\{1,\dots,n\}$. As aforementioned, the notion of distance we use between distributions $\D_1$, $\D_2$ is their \emph{total variation distance}, defined as
  \[ \totalvardist{\D_1}{\D_2}\eqdef\max_{S\subseteq[n]}\left( D_1(S)-D_2(S)\right) = \frac{1}{2}\sum_{i\in[n]}\abs{\D_1(i)-\D_2(i)}.\]
Recall that any property $\mathcal{P}$ can equivalently be seen as the subset of distributions that satisfy it; in particular, the distance $\totalvardist{\D}{\mathcal{P}}$ from some $\D$ to $\mathcal{P}$ is the minimum distance to any distribution in this subset, $\min_{ \D^\prime\in\mathcal{P} } \totalvardist{\D}{\D^\prime}$.\smallskip

\noindent Testing algorithms for distributions over $[n]$ are defined as follows\footnotemark:
\begin{definition}\label{def:testing:alg}
  Fix any property $\mathcal{P}$ of distributions, and let $\ORACLE_D$ be an oracle providing some type of access to $D$. A \emph{$q$-query testing algorithm for $\mathcal{P}$} is a randomized algorithm $\Tester$ which takes as input $n$, $\eps\in(0,1]$, as well as access to $\ORACLE_D$. After making at most $q(\eps,n)$ calls to the oracle, \Tester outputs either \accept or \reject, such that the following holds:
  \begin{itemize}
    \item if $\D\in\mathcal{P}$, \Tester outputs \accept with probability at least $2/3$;
    \item if $\totalvardist{\D}{\mathcal{P}} \geq \eps$, \Tester outputs \reject with probability at least $2/3$.
  \end{itemize}
\end{definition}
We shall also be interested in \emph{tolerant} testers -- \newer{roughly}, algorithms robust to a relaxation of the first item above:
\begin{definition}\label{def:tol:testing:alg}
  Fix property $\mathcal{P}$ and $\ORACLE_D$ as above. A \emph{$q$-query tolerant testing algorithm for $\mathcal{P}$} is a randomized algorithm $\Tester$ which takes as input $n$, $0 \leq \eps_1 < \eps_2 \leq 1$, as well as access to $\ORACLE_D$. After making at most $q(\eps_1,\eps_2,n)$ calls to the oracle, \Tester outputs either \accept or \reject, such that the following holds:
  \begin{itemize}
    \item if $\totalvardist{\D}{\mathcal{P}} \leq \eps_1$, \Tester outputs \accept with probability at least $2/3$;
    \item if $\totalvardist{\D}{\mathcal{P}} \geq \eps_2$, \Tester outputs \reject with probability at least $2/3$.
  \end{itemize}
\end{definition}
Observe in particular that if $\totalvardist{\D}{\mathcal{P}} \in (0,\eps)$ (resp. $\totalvardist{\D}{\mathcal{P}} \in (\eps_1,\eps_2)$), the tester's output can be arbitrary. Furthermore, we stress that the two definitions above only deal with the query complexity, and not the running time. However, it is worth noting that while our lower bounds hold even for such computationally unbounded algorithms, all our upper bounds are achieved by testing algorithms whose running time is polynomial in the number of queries they make.
\ifnum\fulldetails=1
\footnotetext{Note that, as standard in property testing, the threshold $2/3$ is arbitrary: any $1-\delta$ confidence can be achieved at the cost of a multiplicative factor $\log(1/\delta)$ in the query complexity, by repeating the test and outputting the majority vote.}
\else
\footnotetext{Note that, as standard in property testing, the threshold $2/3$ is arbitrary: any $1-\delta$ confidence can be achieved at the cost of a multiplicative factor $\log(1/\delta)$ in the query complexity.}
\fi
  
  \begin{remark}
  We will sometimes refer as a \emph{multiplicatively noisy} $\EVAL_\D$ (or similarly for $\CDFEVAL_\D$) to an evaluation oracle with takes an additional input parameter $\tau > 0$ and returns a value $\hat{d}_i$ within a multiplicative factor $(1+\tau)$ of the true $\D(i)$. \newer{Note however that this notion of noisy oracle does not preserve the two-query simulation of a \pdfsamp oracle by a \cdfsamp one.}
  \end{remark}

Finally, recall the following well-known result on distinguishing biased coins (which can for instance be derived from Eq.~(2.15) and~(2.16) of~\cite{AdellJodra:06}), that shall come in handy in proving our lower bounds:
\begin{fact}\label{fact:fair:biased:coin}
Let $p\in[\eta, 1-\eta]$ for some fixed $\eta > 0$, and suppose $m\leq\frac{c_\eta}{\eps^2}$, with $c_\eta$ a sufficiently small constant
and $\eps < \eta.$ Then,
\[ \totalvardist{ \binomial{m}{p} }{ \binomial{m}{p+\eps} } < \frac{1}{3}. \]
\end{fact}

\noindent We shall make extensive use of Chernoff bounds; for completeness, we state them in Appendix~\ref{appendix:chernoff}.

\section{Uniformity and \newer{identity of} distributions}\label{sec:uniformity:equivalence:identity}
\subsection{Testing}\label{sec:uniformity:equivalence:identity:testing}
In this section, we consider the three following testing problems, each of them a generalization of the previous:
\begin{description}
  \item[Uniformity testing:] given oracle access to $\D$, decide whether $\D=\uniform$ (the uniform distribution on $[n]$) or is far from it;
  \item[Identity testing:] given oracle access to $\D$ and the full description of a fixed $\D^\ast$, decide whether $\D=\D^\ast$ or is far from it; 
  \item[Closeness testing:] given independent oracle accesses to $\D_1$, $\D_2$ (both unknown), decide whether $\D_1=\D_2$ or $\D_1$, $\D_2$ are far from each other. 
\end{description}

\noindent We begin by stating here two results from the literature that transpose straighforwardly in our setting. \new{Observe that since the problem of testing closeness between two unknown distributions $\D_1,\D_2$ in particular encompasses the identity to known $\D^\ast$ testing (and a fortiori the uniformity testing) one, this upper bound automatically applies to these as well.}
\begin{theorem}[Theorem 24 from \cite{RubinfeldServedio:09}]\label{theorem:testing:equivalence:dknown:rs09}
In the query access model, there exists a tester for identity to a known distribution $\D^\ast$ with query complexity $\bigO{\frac{1}{\eps}}$.
\end{theorem}
Note that the tester given in \cite{RubinfeldServedio:09} is neither tolerant nor robust; however, it only uses query access. \cite{CRS12} later adapt this algorithm to give a tester for closeness between two unknown distributions, in a setting which can be seen as ``relaxed'' \pdfsamp access model\footnote{\newer{In the sense that the evaluation oracle, being simulated via another type of oracle, is not only noisy but also allowed to err on a small set of points.}}:
\begin{theorem}[Theorem 12 from \cite{CRS12}]\label{theorem:testing:equivalence:d1:d2:crs12}
In the \pdfsamp access model, there exists a tester for closeness between two unknown distributions $\D_1$, $\D_2$ with sample complexity $\bigO{\frac{1}{\eps}}$.
\end{theorem}
It is worth noting that the algorithm in question is conceptually very simple -- namely, it consists in drawing samples from both distributions and then querying the respective probability mass both distributions put on them, hoping to detect a violation.

\begin{remark}
As mentioned, the setting of the theorem is slightly more general than stated -- indeed, it only assumes ``approximate'' query access to $\D_1$, $\D_2$ (in their terminology, this refers to an evaluation oracle that outputs, on query $x\in[n]$, a good \emph{multiplicative} estimate of $\D_i(x)$, for \emph{most} of the points $x$).
\end{remark}

\paragraph{Lower bound} Getting more efficient testing seems unlikely -- the dependence on $1/\eps$ being ``as good as it gets''. The following result formalizes this, showing that indeed both Theorems~\ref{theorem:testing:equivalence:dknown:rs09} and \ref{theorem:testing:equivalence:d1:d2:crs12} are tight, even for the least challenging task of testing uniformity:
\begin{theorem}[Lower bound for \pdfsamp oracle testers]\label{theorem:testing:uniformity:lb}
In the \pdfsamp access model, any tester for uniformity must have query complexity $\bigOmega{\frac{1}{\eps}}$.
\end{theorem}
Albeit the lower bound above applies only to the \pdfsamp model, one can slightly adapt the proof to get the following improvement:
\begin{theorem}[Lower bound for \cdfsamp oracle testers]\label{theorem:testing:uniformity:lb:cdfquery}
In the \cdfsamp access model, any tester for uniformity must have sample complexity $\bigOmega{\frac{1}{\eps}}$.
\end{theorem}
Albeit the lower bound above applies only to the \pdfsamp model, one can slightly adapt the construction to get the following improvement:
\begin{theorem}[Lower bound for \cdfsamp oracle testers]\label{theorem:testing:uniformity:lb:cdfquery}
In the \cdfsamp access model, any tester for uniformity must have sample complexity $\bigOmega{\frac{1}{\eps}}$.
\end{theorem}
\begin{proof}[Sketch]
Theorem~\ref{theorem:testing:uniformity:lb:cdfquery} directly implies Theorem~\ref{theorem:testing:uniformity:lb}, so we focus on the former. The high-level idea is to trick the algorithm by somehow ``disabling'' the additional flexibility coming from the oracles.

To do so, we start with a distribution that is far from uniform, but easy to recognize when given evaluation queries. We then shuffle its support randomly in such a way that \textsf{(a)} sampling will not, with overwhelming probability, reveal anything, while \textsf{(b)} evaluation queries essentially need to find a needle in a haystack. Note that the choice of the shuffling must be done carefully, as the tester has access to the cumulative distribution function of any \textsf{no}-instance $\D$: in particular, using a random permutation will not work. Indeed, it is crucial for the cumulative distribution function to be as close as the linear function $x\in[n] \mapsto \frac{x}{n}$ as possible; meaning that the set of elements on which $\D$ differs from $\uniform$ had better be a consecutive ``chunk'' (otherwise, looking at the value of the cdf at a uniformly random point would give away the difference with uniform with non-negligible probability: such a point $x$ is likely to have at least a ``perturbed point'' before \emph{and} after it, so that $\sum_{i \leq x} \D(x)\neq \frac{x}{n}$).

Fix any $\eps\in(0,\half]$; for $n\geq \frac{1}{\eps}$, set $m\eqdef(1-\eps)n-1$, and consider testing a distribution $\D$ on $[n]$ which is either \textsf{(a)} the uniform distribution or \textsf{(b)} chosen uniformly at random amongst the family of distributions $(\D_r)_{0\leq r \leq m}$, defined this way: for any offset $0\leq r < m$, $\D_r$ is obtained as follows:\begin{enumerate}
  \item Set $\D(1)=\eps+\frac{1}{n}$, $\D(2)=\dots=\D(\eps n+1)=0$, and $\D(k)=\frac{1}{n}$ for the remaining $m=(1-\eps)n-1$ points;
  \item Shift the whole support (modulo $n$) by adding $r$.
\end{enumerate}
At a high-level, what this does is keeping the ``chunk'' on which the cdf of the \textsf{no}-instance grouped together, and just place it at a uniformly random position; outside this interval, the cdf's are exactly the same, and the only way to detect a difference with \CDFEVAL is to make a query in the ``chunk''. Furthermore, it is not hard to see that any \textsf{no}-instance distribution will be exactly $\eps$-far from uniform, so that any tester $\Tester$ must distinguish between cases \textsf{(a)} and \textsf{(b)} with probability at least 2/3.\medskip

Suppose by contradiction that there exists a tester $\Tester$ making $q=\littleO{\frac{1}{\eps}}$ queries (\newer{without loss of generality}, we can further assume $\Tester$ makes exactly $q$ queries; and that for any $\SAMP$ query, the tester also gets ``for free'' the result of an evaluation query on the sample). Given \pdfsamp access to a $\D=\D_r$ generated as in case \textsf{(b)}, observe first that, since the outputs of the sample queries are independent of the results of the evaluation queries, one can assume that some evaluation queries are performed first, followed by some sample queries, before further evaluation queries (where the evaluation points may depend arbitrarily on the sample query results) are made. That is, we subdivide the queries in 3: first, $q_1$ consecutive $\EVAL$ queries, then a sequence of $q_2$ $\SAMP$ queries, and finally $q_3$ $\EVAL$ queries. Define the following ``bad'' events:
\begin{itemize}
  \item $E_1$: one of the first $q_1$ evaluation queries falls outside the set $G\eqdef\left\{\newer{\eps n+2+r}, \dots, n+r\right\} \mod n$;
  \item $E_2$: one of the $q_2$ sampling queries returns a sample outside $G$, conditioned on $\overline{E_1}$; 
  \item $E_3$: one of the $q_3$ evaluation queries is on a point outside $G$, conditioned on $\overline{E_1}\cap\overline{E_2}$.
\end{itemize}
It is clear that, conditioned on $\overline{E_1}\cap\overline{E_2}\cap\overline{E_3}$, all the tester sees is exactly what its view would have been in case \textsf{(a)} (probabilities equal to $\frac{1}{n}$ for any $\EVAL$ query, and uniform sample from $G$ for any $\SAMP$ one). It is thus sufficient to show that $\probaOf{ \overline{E_1}\cap\overline{E_2}\cap\overline{E_3} } = 1-\littleO{1}$.
\begin{itemize}
  \item As $r$ is chosen uniformly at random, $\probaOf{ E_1 } \leq q_1\frac{n-m}{n}=q_1(\eps+\frac{1}{n})$;
  \item since $D(G)=\frac{m}{n} = 1-\eps -\frac{1}{n} \geq 1-2\eps$, $\probaOf{ E_2 } \leq 1 - (1-2\eps)^{q_2}$;
  \item finally, $\probaOf{ E_3 } \leq q_3(\eps+\frac{1}{n})$;
\end{itemize}
we therefore have $\probaOf{ E_1 \cup E_2\cup E_3 } \leq (q_1+q_3)(\eps+\frac{1}{n}) + 1 - (1-2\eps)^{q_2} = \bigO{q\eps} = \littleO{1}$, as claimed.
\end{proof}

\subsection{Tolerant testing}\label{sec:uniformity:equivalence:identity:testing:tolerant}
In this section, we describe tolerant testing algorithms for the three problems of uniformity, identity and closeness; note that by a standard reduction (see Parnas et al. (\cite{PRR:2006}, Section 3.1), this is equivalent to estimating the distance between the corresponding distributions. As \new{hinted in the introduction, our algorithm relies on a general estimation approach that will be illustrated further in Section~\ref{sec:supportsize:entropy}, and which constitutes a fundamental feature of the \pdfsamp oracle}: namely, the ability to estimate \newer{cheaply} quantities of the form
$\shortexpect_{i\sim \D}\left[ \Phi(i,\D(i)) \right]$ for any \emph{bounded} function $\Phi$.

%
\begin{theorem}\label{theorem:tolerant:tester:uniform}
In the \pdfsamp access model, there exists a tolerant tester for uniformity with query complexity $\bigO{\frac{1}{(\eps_2-\eps_1)^2}}$.
\end{theorem}
\begin{proof}
  We describe such a tester $\mathcal{T}$; as it will start by estimating the quantity $2\totalvardist{\D}{\uniform}$ up to some additive $\gamma\eqdef{\eps_2-\eps_1}$ (and then accept if and only if its estimate $\hat{d}$ is at most $2\eps_1+\gamma = \eps_1+\eps_2$).
  
\noindent In order to approximate this quantity, observe that\footnote{Note that dividing by $\D(i)$ is ``legal'', since if $\D(i)=0$ for some $i\in[n]$, this point will never be sampled, and thus no division by $0$ will ever occur.}
  \begin{align}
  \totalvardist{\D}{\uniform} &= \frac{1}{2}\sum_{i\in[n]} \abs{\D(i)-\frac{1}{n}} = \sum_{i:\D(i)> \frac{1}{n} } \left(\D(i)-\frac{1}{n}\right) = \sum_{i:\D(i)> \frac{1}{n}} \left(1-\frac{1}{n\D(i)}\right)\cdot \D(i) \notag\\
  &= \shortexpect_{i\sim \D}\left[ \left(1 - \frac{1}{n\D(i)}\right)\indic{\D(i)> \frac{1}{n}} \right] \label{eq:trick:dtv:as:expectation}
  \end{align}
  where $\indicSet{E}$ for the indicator function of set (or event) $E$; thus, $\mathcal{T}$ only has to do get an empirical estimate of \newer{this expected value}, which can be done by taking $m=\bigO{{1}/{(\eps_2-\eps_1)^2}}$ samples $s_i$ from $\D$, querying $\D(s_i)$ and computing $X_i = \left(1 - \frac{1}{nD(s_i)}\right)\indic{D(s_i)> \frac{1}{n}}$ (cf. Alg.~\ref{algo:compute:l1:distance:samp+eval}).
  \begin{algorithm}[h!]
    \begin{algorithmic}
      \Require $\SAMP_{\D}$ and $\EVAL_{\D}$ oracle access, parameters $0\leq \eps_1 < \eps_2$
      \State Set $m\eqdef\bigTheta{\frac{1}{\gamma^2}}$, where $\gamma\eqdef\frac{\eps_2-\eps_1}{2}$.
      \State Draw $s_1,\dots,s_m$ from $\D$
      \For{$i=1 \textbf{ to } m$}
        \State With \EVAL, get $X_i \eqdef \left(1 - \frac{1}{nD(s_i)}\right)\indic{D(s_i)> \frac{1}{n}}$
      \EndFor
      \State Compute $\hat{d}\eqdef \frac{1}{m} \sum_{i=1}^m X_i$.
      \If{ $\hat{d}\leq \frac{\eps_1+\eps_2}{2}$} \State \Return\accept  \Else \State \Return\reject \EndIf
    \end{algorithmic}
    \caption{\label{algo:compute:l1:distance:samp+eval}Tester $\mathcal{T}$: \textsc{Estimate-$L_1$} }
  \end{algorithm}

  \paragraph*{Analysis} Define the random variable $X_i$ as above; from Eq.\eqref{eq:trick:dtv:as:expectation}, we can write its expectation as
    \begin{equation}
    \expect{ X_i } = \sum_{k=1}^n D(k)\abs{1-\frac{1}{n D(k)}}\indic{D(k)> \frac{1}{n}} = \totalvardist{\D}{\uniform}.
    \end{equation}
    Since the $X_i$'s are independent and take value in $[0,1]$, an additive Chernoff bound ensures that
    \begin{equation}\label{eq:chernoff:dhat}
      \probaOf{ \abs{\hat{d} - \totalvardist{\D}{\uniform}} \geq \gamma  } \leq 2e^{-2\gamma^2 m}
    \end{equation}
    which is at most $1/3$ by our choice of $m$. Conditioning from now on on the event {$\abs{\hat{d} - \totalvardist{\D}{\uniform} } < \gamma$}:
    \begin{itemize}
      \item if $\totalvardist{\D}{\uniform}\leq \eps_1$, then $\hat{d}\leq \eps_1+\gamma$, and $\mathcal{T}$ outputs \accept;
      \item if $\totalvardist{\D}{\uniform} > \eps_2$, then $\hat{d} > \eps_2-\gamma$, and $\mathcal{T}$ outputs \reject.
    \end{itemize}
    Furthermore, the algorithm makes $m$ \SAMP queries, and $m$ \EVAL queries.
\end{proof}

\begin{remark}\label{remark:tolerant:tester:uniform:only:eval}
  Note that we can also do it with \EVAL queries only (same query complexity), by internally drawing uniform samples: indeed,
  \[
  2\totalvardist{\D}{\uniform} = \sum_{i\in[n]} \abs{\D(i)-\frac{1}{n}} = \sum_{i\in[n]} \abs{n\D(i)-1}\cdot \frac{1}{n} = 2\shortexpect_{x\sim \uniform}\left[ \abs{n\D(x)-1}\indic{\frac{1}{n} > \D(x)} \right]
  \]
  This also applies to the first corollary below, as long as the known distribution is efficiently samplable by the algorithm.
\end{remark}
Indeed, the proof above can be easily extended to other distributions than uniform, and even to the case of two unknown distributions:
\begin{corollary}\label{coro:tolerant:tester:known:d}
In the \pdfsamp access model, there exists a tolerant tester for \newer{identity} to a known distribution with query complexity $\bigO{\frac{1}{(\eps_2-\eps_1)^2}}$.
\end{corollary}

\begin{corollary}\label{coro:tolerant:tester:unknown:d1:d2}
In the \pdfsamp access model, there exists a tolerant tester for \newer{closeness} between two unknown distributions with query complexity $\bigO{\frac{1}{(\eps_2-\eps_1)^2}}$. As noted in the next subsection, this is optimal (up to constants).
\end{corollary}
\ifnum\fulldetails=1
  \noindent Interestingly, this tester  can be made \newer{robust to multiplicative noise}, i.e. can be shown to work even when the answers to the \EVAL queries are only accurate up to a factor $(1+\gamma)$ for $\gamma>0$: it suffices to set $\gamma=\eps/2$, getting on each point $\hat{D}(i)\in[(1+\gamma)^{-1},1+\gamma]\D(i)$, and work with $X_i = \left(1 - \D^{\ast}(s_i)/\hat{D}(s_i)\right)\indic{\hat{D}(s_i)> \D^{\ast}(s_i)}$ and estimate the expectation up to $\pm\gamma$ (or, for \newer{closeness} between two unknown distributions, setting $\gamma=\eps/4$).
\else
  \noindent Interestingly, this tester  can be made \newer{robust to multiplicative noise}, i.e. can be shown to work even when the answers to the \EVAL queries are only accurate up to a factor $(1+\gamma)$ for $\gamma>0$.
\fi
  

\subsubsection{Lower bound}
In this subsection, we show that the upper bounds of Lemma~\ref{theorem:tolerant:tester:uniform}, Corollaries~\ref{coro:tolerant:tester:known:d} and \ref{coro:tolerant:tester:unknown:d1:d2} are tight.

\begin{theorem}\label{theorem:tolerant:tester:uniform:lb} In the \pdfsamp access model, performing $(\eps_1, \eps_2)$-testing for uniformity requires sample complexity $\bigOmega{\frac{1}{(\eps_2-\eps_1)^2}}$ (the bound holds even when only asking $\eps_1$ to be $\bigOmega{1}$).
\end{theorem}
 \begin{proof}
The overall idea lies on a reduction from distinguishing between two types of biased coins to tolerant testing for uniformity. In more detail, given access to samples from a fixed coin (promised to be of one of these two types), we define a probability distribution as follows: the domain $[n]$ is randomly partitioned into $K=1/\eps^2$ pairs of buckets, each bucket having same number of elements; the distribution is uniform within each bucket, and the two buckets of each pair are balanced to have total weight $2/K$. However, within each pair of buckets $(A,B)$, the probability mass is divided according to a coin toss (performed ``on-the-fly'' when a query is made by the tolerant tester), so that either \textsf{(a)} $D(A)=(1+\alpha)/K$ and $D(B)=(1-\alpha)/K$, or \textsf{(b)} $D(A)=D(B)=1/K$. Depending on whether the coin used for this choice is fair or $(\half+\eps)$)biased, the resulting distribution will (with high probability) have different distance from uniformity -- sufficiently for a tolerant tester to distinguish between the two cases.

  \paragraph*{Construction} We start by defining the instances of distributions we shall consider. Fix any $\eps\in(0,\frac{1}{100})$; without loss of generality, assume $n$ is even, and $n \gg 1/\eps$. Define $\alpha = 2/(1+\eps)$, $K= 1/\eps^2$, $p^+ = (1+\eps)/2$ and $p^- = (1+20\eps)/2$, and consider the family of distributions $\mathcal{D}^+$ (resp. $\mathcal{D}^-$) defined by the following construction:
  \begin{itemize}
    \item pick \newer{uniformly at random} a partition\footnote{For convenience, it will be easier to think of the $A_i$'s and $B_i$'s as consecutive intervals, the first ones covering $[\frac n 2]$ while the former cover $[n]\setminus[\frac n 2]$ (see Fig.~\ref{fig:construction:lb:tolerant:uniformity:D:pm}).} of $[n]$ in $2K$ sets of size $n/(2K)$ $A_1,\dots, A_K, B_1,\dots, B_K$;
    \item for all $k\in[K]$, draw independently at random $X_k\sim\bernoulli{p^+}$ (resp. $X_k\sim\bernoulli{p^-}$), and set for all $x\in A_k$, $y\in B_k$
    \[
    \D^+(x)=\begin{cases}
      \frac{1+\alpha}{n} & \text{ if } X_i = 1 \\
      \frac{1}{n} & \text{ o.w.}
    \end{cases}
    \qquad\text{and}\qquad
    \D^+(y)=\begin{cases}
      \frac{1-\alpha}{n} & \text{ if } X_i = 1 \\
      \frac{1}{n} & \text{ o.w.}
    \end{cases}
    \]
  \end{itemize}
  (the pairing between $A_k$ and $B_k$ ensures the final measure indeed sums to one).
  Regardless of the choice of the initial partition, but with fluctuations over the random coin tosses $X_1,\dots, X_k$, we have that the total variation distance between a distribution $\D^+\in\mathcal{D}^+$ (resp. $\D^-\in\mathcal{D}^-$) and uniform is on expectation what we aimed for:
  \begin{align*}
  \expect{ \totalvardist{\D^+}{\uniform} } &= \half\cdot 2\cdot\sum_{k=1}^K \frac{n}{2K}\cdot\frac{\alpha}{n}p^+ = \half\alpha p^+ = \half \\
  \expect{ \totalvardist{\D^{-\vphantom{+}}}{\uniform} } &= \half p^-\alpha = \frac{1+20\eps}{1+\eps}\cdot\frac{1}{2} > \frac{1}{2} + 7\eps
  \end{align*}
  and with an additive Chernoff bound 
   on the sum of $1/\eps^2$ i.i.d. choices for the $X_k$'s, we have that for $(\D^+, \D^-)$: for any choice of the initial partition $\Pi=(A_k,B_k)_{k\in[K]}$, with probability at least $99/100$,
  \begin{align*}
  \totalvardist{\D_\Pi^+}{\uniform} &< \frac{1}{2} + 3\eps \\
  \totalvardist{\D_\Pi^-}{\uniform} &> \frac{1}{2} + 4\eps
  \end{align*}
  where by $\D_\Pi^\pm$ we denote the distribution defined as above, but fixing the partition for the initial step to be $\Pi$. We will further implicitly condition on this event happening; any tester for uniformity called with $(\eps^\prime, \eps^\prime+c\eps)$ must therefore distinguish between $\D^+$ and $\D^-$. Suppose we have such a tester $\mathcal{T}$, with (\newer{without loss of generality}) exact sample complexity $q=q(\eps)=\littleO{\frac{1}{\eps^2}}$.
  
  
\begin{figure}[!ht]\centering
  \begin{tikzpicture}[x=6pt, y=2pt]
  \pgfmathsetmacro{\xmax}{ 70 }
  \pgfmathsetmacro{\ymax}{40}
  \pgfmathsetmacro{\alphavalue}{0.20}
    \draw [<->] (0,\ymax) node[above] {$\D^+(i)$} -- (0,0) -- (\xmax,0)  node[right] {$i$};

  \pgfmathsetmacro{\rbuckets}{32} 
  \pgfmathparse{{int(floor(\rbuckets/2)-1))}}\let\rbucketshalf\pgfmathresult
  \pgfmathsetmacro{\buckwidth}{(\xmax/\rbuckets)}
  
      \node [left] at (0,{\ymax*(1/2+\alphavalue)}) {$\frac{1+\alpha}{n}$};
      \node [left] at (0,{\ymax*(1/2-\alphavalue)}) {$\frac{1-\alpha}{n}$};
      \node [below] at (\xmax/2,-3) {$\frac{n}{2}$};
      \node [below] at (\xmax,-3) {$n$};
      \node [below] at (0,-3) {$1$};
      \draw[thin, black, dotted]  (0,{\ymax*(1/2)}) -- (\xmax,{\ymax*(1/2)} );

    \foreach \i in {1,...,3}{
      \node[below] at ({(\i-0.5)*\buckwidth},0) {\scriptsize $A_{\i}$};
      \node[below] at ({\xmax-(\i-0.5)*\buckwidth},0) {\scriptsize $B_{\i}$};
    };
    \foreach \i in {\rbucketshalf+1}{
      \node[below] at ({(\i-0.75)*\buckwidth},0) {\scriptsize $A_{K}$};
      \node[below] at ({\xmax-(\i-0.75)*\buckwidth},0) {\scriptsize $B_{K}$};
    };
    
    \draw[thin, dotted] \foreach \i in {1,...,\rbuckets}{
      ({\i*\buckwidth},0) -- ({\i*\buckwidth},\ymax)
    };
    \draw[thin] ({\xmax/2},0) -- ({\xmax/2},\ymax);
    
        \foreach \i in {0,...,\rbucketshalf}{
          \pgfmathsetmacro{\cointoss}{round(random)};
          \pgfmathsetmacro{\fstHigh}{{\ymax*(1/2+\cointoss*\alphavalue)}};
          \pgfmathsetmacro{\sndHigh}{{\ymax*(1/2-\cointoss*\alphavalue)}};
          
          \pgfmathtruncatemacro{\iIsEven}{int(\cointoss))}
          \ifnum\iIsEven=0
            \xdef\linestyle{red};
          \else
            \xdef\linestyle{blue};
          \fi
          
          \draw[thin,\linestyle]  ({\i*\buckwidth},{\fstHigh}) -- ({(\i+1)*\buckwidth},{\fstHigh});
          \draw[thin,\linestyle]  ({\xmax-\i*\buckwidth},{\sndHigh}) -- ({\xmax-(\i+1)*\buckwidth},{\sndHigh});
        };
  \end{tikzpicture}\caption{\label{fig:construction:lb:tolerant:uniformity:D:pm}The \textsf{yes}-instance $\D^+$ (for a fixed $\Pi$, taken to be consecutive intervals).}
\end{figure}
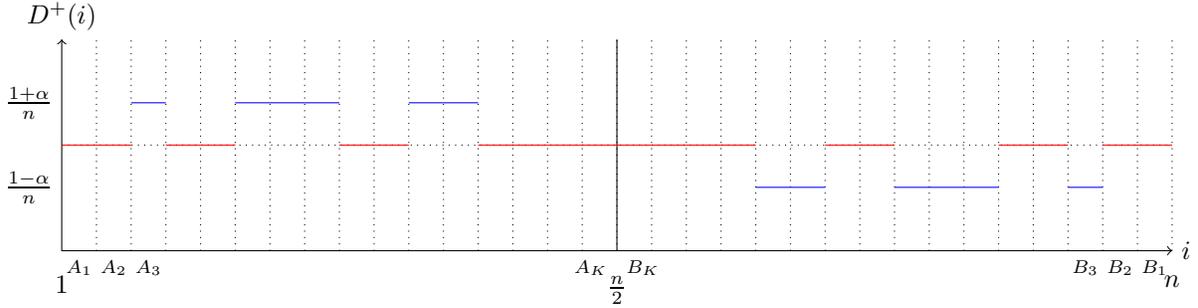
  
  \paragraph*{Reduction} We will reduce the problem of distinguishing between $\textsf{(a)}$ a $\frac{p^+\alpha}{2}$ and $\textsf{(b)}$ a $\frac{p^-\alpha}{2}$ biased coin to telling $\D^+$ and $\D^-$ apart.
  
Given $\SAMP_\text{coin}$ access to i.i.d. coin tosses coming either from one of those two situations, define a distinguisher \Algo as follows:
  \begin{itemize}
    \item choose \newer{uniformly at random} a partition $\Pi=(A^0_k,A^1_k)_{k\in[K]}$ of $[n]$; for convenience, for any $i\in[n]$, we shall write $\pi(i)$ for the index $k\in[K]$ such that $i\in A^0_k\cup A^1_k$, and $\varsigma(i)\in\{0,1\}$ for the part in which it belongs -- so that $i\in A^{\varsigma(i)}_{\pi(i)}$ for all $i$;
    \item run $\mathcal{T}$, maintaining a set $C$ of triples\footnote{Abusing the notation, we will sometimes write ``$k\in C$'' for ``there is a triple in $C$ with first component $k$''.} $(k, \D^0_k, \D^1_k)$ (initially empty), containing the information about the $(A^0_k,A^1_k)$ for which the probabilities have already be decided;
    \item $\EVAL$: whenever asked an evaluation query on some $i\in[n]$:
      \begin{itemize}
        \item if $\pi(i)\in C$, return $\D^{\varsigma(i)}_{\pi(i)}$;
        \item otherwise, let $k=\pi(i)$; ask a fresh sample $b_k$ from $\SAMP_\text{coin}$ and draw a uniform random bit $b_k^\prime$; set
    \[
    (\D^0_k, \D^1_k)=\begin{cases}
      (\frac{1}{n},\frac{1}{n}) & \text{ if } b_k = 0 \\
      (\frac{1+\alpha}{n},\frac{1-\alpha}{n}) & \text{ if } b_k = 1, b^\prime_k = 1 \tag{``Choosing the profile''}\\
      (\frac{1-\alpha}{n},\frac{1+\alpha}{n}) & \text{ if } b_k = 1, b^\prime_k = 0
    \end{cases}\]  
    then add $(k,\D^0_k, \D^1_k)$  to $C$; and return $\D^{\varsigma(i)}_{k}$.
      \end{itemize}
    \item $\SAMP$: whenever asked a sample: let $\gamma=\frac{n}{2K}\sum_{k\in C} d_k$ the current probability mass of the ``committed points''; observe that the distribution $\D_C$ induced by the $d_k$'s on $\setOfSuchThat{i\in[n]}{\pi(i)\in C}$ is fully known by \Algo;
    \begin{itemize}
      \item with probability $\gamma$, \Algo draws $i\sim \D_C$ and returns it;
      \item otherwise, \Algo draws $k\sim\uniformOn{[K]\setminus C}$. As before, it gets $b_k$ from $\SAMP_\text{coin}$ and a uniform random bit $b_k^\prime$; gets $(\D^0_k, \D^1_k)$ as in the \EVAL case, commits to it as above by $(k,\D^0_k, \D^1_k)$ to $C$. Finally, it draws a random sample $i$ from the piecewise constant distribution induced by  $(\D^0_k, \D^1_k)$ on $A^0_k\cup A^1_k$, where each $j\in A^0_k$ (resp. $A^1_k$) has equal probability mass $\D^0_k\cdot\frac{n}{2K}$ (resp. $\D^1_k\cdot\frac{n}{2K}$), and returns $i$.
    \end{itemize}
  \end{itemize}
  Observe that $\Algo$ makes at most $q$ queries to $\SAMP_\text{coin}$; provided we can argue that \Algo answers $\mathcal{T}$'s queries consistently to what a corresponding $\D^\pm$ (depending on whether we are in case $\textsf{(a)}$ or $\textsf{(b)}$) would look like, we can conclude.
  
  \noindent This is the case, as \textsf{(i)} $\Algo$ is always consistent with what its previous answers induce on the distribution (because of the maintaining of the set $C$); \textsf{(ii)} any \EVAL query on a new point exactly simulates the ``on-the-fly'' construction of a $\D^\pm$; and any \SAMP query is either consistent with the part of $\D^\pm$ already built, or in case of a new point gets a sample exactly distributed according to the $\D^\pm$ built ``on-the-fly''; this is because in any $\D^\pm$, every $A_k\cup B_k$ has same probability mass $1/(2K)$; therefore, in order to get one sample, tossing $K$ i.i.d. coins to decide the ``profiles'' of \emph{every} $A_k\cup B_k$ before sampling from the overall support $[n]$ is equivalent to first choosing uniformly at random a particular $S=A_k\cup B_k$, tossing one coin to decide \emph{only its particular profile}, and then drawing a point accordingly from $S$.
  
  In other terms, \Algo will distinguish, with only $\littleO{1/\eps^2}$ i.i.d. samples, between cases $\textsf{(a)}$ ($\frac{1}{2}$-biased coin) and $\textsf{(b)}$ ($\frac{1}{2}+\bigOmega{\eps}$)-biased coin with probability at least $6/10$ -- task which, for $\eps$ sufficiently small, is known to require $\bigOmega{1/\eps^2}$ samples (cf. Fact~\ref{fact:fair:biased:coin}), thus leading to a contradiction.
\end{proof}


\section{Entropy and support size}\label{sec:supportsize:entropy}
\subsection{Additive and multiplicative estimations of entropy}\label{sec:supportsize:entropy:support}

In this section, we describe simple algorithms to perform additive and multiplicative estimation (which in turns straightforwardly implies tolerant testing) of the \emph{entropy} $H(\D)$ of the unknown distribution $\D$, defined as 
\[
  H(\D)\eqdef -\sum_{i\in[n]} \D(i)\log\D(i) \in[0,\log n]
\]
We remark that Batu et al. (\cite{BDK+:05}, Theorem 14) gives a similar algorithm, based on essentially the same approach but relying on a Chebyshev bound, yielding a $(1+\gamma)$-multiplicative approximation algorithm for entropy with sample complexity $\bigO{(1+\gamma)^2\log^2 n/\gamma^2h^2}$, given a lower bound $h>0$ on $H(\D)$. 

Guha et al. (\cite{GMV:05}, Theorem 5.2) then refined their result, using as above a threshold for the estimation along with a multiplicative Chernoff bound to get the sample complexity down to $\bigO{\log n/\gamma^2 h}$~--~thus matching the $\bigOmega{{\log n}/{\gamma(2+\gamma) h}}$ lower bound of \cite{BDK+:05} (Theorem 18); we recall their results for multiplicative estimation of the entropy below\footnote{In particular, note that translating their lower bound for additive estimation implies that the dependence on $n$ of our algorithm is tight.}.

\begin{theorem}[Upper bound {[\cite{GMV:05}, Theorem 5.2]}]\label{theorem:entropy:estimation:multiplicative}
Fix $\gamma > 0$. In the \pdfsamp access model, there exists an algorithm that, given a parameter $h >0$ and the promise that $H(\D) \geq h$, estimates the entropy within a multiplicative $(1+\gamma)$ factor, with sample complexity $\bigTheta{\frac{\log n }{\gamma^2 h}}$.
\end{theorem}
\begin{theorem}[Lower bound {[\cite{BDK+:05}, Theorem 18]}]\label{theorem:entropy:estimation:multiplicative:lb}
Fix $\gamma > 0$. In the \pdfsamp access model, any algorithm that, given a parameter $h >0$ and the promise that $H(\D) =\bigOmega{h}$, estimates the entropy within a multiplicative $(1+\gamma)$ factor must have sample complexity $\bigOmega{\frac{\log n}{\gamma(2+\gamma) h}}$.
\end{theorem}

 Observe that the additive bound we give (based on a different cutoff threshold), however, still performs better in many cases, e.g. $\Delta=\gamma h > 1$ and $h > 1$; and does not require any \textit{a priori} knowledge on a lower bound $h > 0$. Moreover, we believe that this constitutes a good illustration of the more general technique used, and a good example of what the \pdfsamp model allows: approximation of quantities of the form $\shortexpect_{i\sim \D}\left[ \Phi(i,\D(i)) \right]$, where $\Phi$ is any \emph{bounded} function of both an element of the domain and its probability mass under the distribution $\D$.

\paragraph{Additive estimate}
The overall idea is to observe that for a distribution $\D$, the entropy $H(\D)$ can be rewritten as
\begin{equation}
  H(\D) = \sum_{x\in[n]} \D(x)\log\frac{1}{\D(x)} = \shortexpect_{x\sim \D}\left[ \log\frac{1}{\D(x)} \right]
\end{equation}
The quantity $\log\frac{1}{\D(x)}$ cannot be easily upperbounded, which we need  for concentration results. However, recalling that the function {$x\mapsto x\log(1/x)$} is increasing for $x\in(0,\frac{1}{e})$ (and has limit 0 when $x\to 0^+$), one can refine the above identity as follows: for any \emph{cutoff threshold} $\tau \in(0,\frac{1}{e})$, write 
\begin{equation}
  H(\D) = \sum_{x:\D(x) \geq \tau} \D(x)\log\frac{1}{\D(x)} + \sum_{x:\D(x) < \tau} \D(x)\log\frac{1}{\D(x)} 
\end{equation}
so that
\begin{align*}
  H(\D) \geq \sum_{x:\D(x) \geq \tau} \D(x)\log\frac{1}{\D(x)} &\geq H(\D) - \sum_{x:\D(x) < \tau} \D(x)\log\frac{1}{\D(x)}  \\
  &\geq H(\D) - n\cdot\tau\log\frac{1}{\tau}
\end{align*}
Without loss of generality, assume $\frac{\Delta}{n} < \half$. Fix $\tau\eqdef\frac{\frac{\Delta}{n}}{10\log \frac{n}{\Delta}}$, so that $n\cdot\tau\log\frac{1}{\tau}\leq \frac{\Delta}{2}$; and set 
\[
  \varphi\colon y \mapsto \log\frac{1}{y}\indic{y \geq \tau}
\]
Then, the above discussion gives us
\begin{equation}
  H(\D) \geq \shortexpect_{x\sim \D}[ \varphi(\D(x)) ] \geq H(\D) - \frac{\Delta}{2}
\end{equation}
and getting an additive $\Delta/2$-approximation of $\shortexpect_{x\sim \D}[ \varphi(\D(x))]$ is enough for estimating $H(\D)$ within $\pm\Delta$; further, we now have
\begin{equation}
  0\leq \varphi(\D(x)) \leq \log\frac{1}{\tau} \sim \log\frac{n}{\Delta} \text{ a.s.}
\end{equation}
so using an additive Chernoff bound, taking $m=\bigTheta{\frac{\log^2\frac{n}{\Delta}}{\Delta^2}}$ samples $x_1,\dots,x_m$ from $\SAMP_D$ and computing the quantities $\varphi(\D(x_i))$ using $\EVAL_\D$ implies
\[
  \probaOf{ \abs{  \frac{1}{m}\sum_{i=1}^m\varphi(\D(x_i)) - \shortexpect_{x\sim \D}[ \varphi(\D(x)) } ]\geq \frac{\Delta}{2} } \leq 2e^{ -\frac{\Delta^2 m}{\log^2\frac{1}{\tau}} } \leq \frac{1}{3}
\]
\noindent This leads to the following theorem:
\begin{theorem}\label{theorem:entropy:estimation:additive}
In the \pdfsamp access model, there exists an algorithm estimating the entropy up to an additive $\Delta$, with sample complexity $\bigTheta{\frac{\log^2\frac{n}{\Delta}}{\Delta^2}}$.
\end{theorem}
or, in terms of tolerant testing:
\begin{corollary}\label{corollary:tolerant:tester:entropy}
In the \pdfsamp access model, there exists an $(\Delta_1, \Delta_2)$-tolerant tester for entropy with sample complexity $\tildeTheta{\frac{\log^2 n}{(\Delta_1-\Delta_2)^2}}$.
\end{corollary}
\begin{proof}
  We describe such a $\mathcal{T}$ in Algorithm~\ref{algo:compute:additive:entropy:samp+eval}; the claimed query complexity is straighforward.
  \vspace{-2ex}
  \begin{algorithm}[h!]
    \begin{algorithmic}
      \Require $\SAMP_{\D}$ and $\EVAL_{\D}$ oracle access, parameters $0\leq \Delta \leq \frac{n}{2}$
      \Ensure Outputs $\hat{H}$ s.t. w.p. at least $2/3$, $\hat{H}\in[H(\D)-\Delta, H(\D)+\Delta/2]$
      \State Set $\tau\eqdef\frac{\frac{\Delta}{n}}{10\log \frac{n}{\Delta}}$ and $m=\clg{\frac{\ln 6}{\Delta^2}\log^2\frac{1}{\tau}}$.
      \State Draw $s_1,\dots,s_m$ from $\D$
      \For{$i=1 \textbf{ to } m$}
        \State With \EVAL, get $X_i \eqdef \log\frac{1}{\D(s_i)}\indic{\D(s_i) \geq \tau}$
      \EndFor
      \State\Return $\hat{H}\eqdef \frac{1}{m} \sum_{i=1}^m X_i$
    \end{algorithmic}
    \caption{\label{algo:compute:additive:entropy:samp+eval}Tester $\mathcal{T}$: \textsc{Estimate-Entropy} }
  \end{algorithm}
\end{proof}
\ifnum\fulldetails=1
  \begin{remark}
    The tester above can easily be adapted to be made multiplicatively robust; indeed, it is enough that the \EVAL oracle only provides $(1+\gamma)$-accurate estimates $\hat{D}(i)$ of the probabilities $\D(i)$, where $\gamma$ is chosen to be $\gamma\eqdef \min(2^{\Delta/3}-1, 1)$
    so that the algorithm will output \newer{with high probability} an additive $(\Delta/2)$-estimate of a quantity
    \[
    H(\D) \geq \shortexpect_{x\sim \D}\left[ \widehat{\varphi}(x) \right] \geq \sum_{x:\D(x) \geq (1+\gamma)\tau} \D(x)\log\frac{1}{\D(x)} - \log(1+\gamma) 
    \geq  H(\D) + n\underbrace{\cdot(1+\gamma)\tau\log(1+\gamma)\tau}_{\geq -2\tau\log\frac{1}{2\tau}} - \frac{\Delta}{3}
    \vspace{-\baselineskip}\]
    and taking for instance $\tau\eqdef\frac{\frac{\Delta}{n}}{30\log \frac{n}{\Delta}}$ ensures the \newer{right-hand-side} is at least $H(\D)-\frac{\Delta}{6}-\frac{\Delta}{3} = H(\D)-\frac{\Delta}{2}$.
  \end{remark}
\else
  \begin{remark}
    The tester above can easily be adapted to be made multiplicatively robust; indeed, it is enough that the \EVAL oracle only provides $(1+\gamma)$-accurate estimates $\hat{D}(i)$ of the probabilities $\D(i)$, where $\gamma$ is set to $\min(2^{\Delta/3}-1, 1)$.
  \end{remark}
\fi

\subsection{Additive estimation of entropy for monotone distributions}\label{sec:supportsize:entropy:monotone}
In the previous section, we saw how to obtain an additive estimate of the entropy of the unknown distribution, using essentially $\bigO{\log^2 n}$ sampling and evaluation queries; moreover, this dependence on $n$ is optimal. However, one may wonder if, by taking advantage of \emph{cumulative} queries, it becomes possible to obtain a better query complexity. \newer{We partially answer this question, focusing on a particular class of distributions for which the \cdfsamp query access seems particularly well-suited: namely the class of \emph{monotone} distributions}\footnote{Recall that a distribution $\D$ \newer{over a totally ordered domain} is said to be monotone if for all $i\in[n-1]$ $\D(i)\geq \D(i+1)$}.

Before describing how this assumption can be leveraged to obtain an exponential improvement in the sample complexity for \cdfsamp query algorithms, we first show that given only \emph{\pdfsamp} access to a distribution promised to be $\littleO{1}$-close to monotone, no such speedup can hold. By establishing (see Remark~\ref{remark:entropy:estimation:additive:monotone:close}) that the savings obtained for (close to) monotone distributions are only possible with \cdfsamp access, this will yield a separation between the two oracles, proving the latter is strictly more powerful.

\subsubsection{Lower bound for \pdfsamp oracles}

\begin{theorem}\label{theorem:entropy:estimation:additive:monotone:lb:pdf}
In the \pdfsamp access model, any algorithm that estimates the entropy of distributions $\bigO{1/\log n}$-close to monotone even to an additive constant must make $\bigOmega{\log n}$ queries to the oracle.
\end{theorem}
\begin{proof}
We will define two families of distributions, $\mathcal{D}_1$ and $\mathcal{D}_2$, such that for any two $D_1$, $D_2$ drawn uniformly at random from $\mathcal{D}_1$ and $\mathcal{D}_2$:
\begin{enumerate}
  \item $D_1$ and $D_2$ are $(2/\log n)$-close to monotone;
  \item $\abs{H(D_1)-H(D_2)} = 1/4$;
  \item no algorithm making $\littleO{\log n}$ queries to a \pdfsamp oracle can distinguish between $D_1$ and $D_2$ with constant probability.
\end{enumerate}
In more detail, the families are defined by the following process: for $K_n\eqdef n^{1/4}$, $\ell_n\eqdef\log n$ and $\gamma_n\eqdef 1/\log n$,
\begin{itemize}
  \item Draw a subset $S\subset\{2,\dots,n\}$ of size $\ell_n$ uniformly at random;
  \item Set $D_1(1)=1-\gamma_n$, and $D_1(i)=\gamma_n/\ell_n=1/\log^2 n$ for all $i\in S$.
\end{itemize}
($D_2$ is obtained similarly, but with a subset $S$ of size $K_n\ell_n= n^{1/4}\log n$ and $D_2(i)=\gamma_n/(\ell_n K_n)$)
Roughly, both distributions have a very heavy first element (whose role is to ``disable'' sampling queries by hogging them with high probability), and then a random subset of size respectively logarithmic or polynomial, on which they are uniform. To determine whether a distribution is drawn from $\mathcal{D}_1$ or $\mathcal{D}_2$, intuitively a testing algorithm has to find a point $i > 1$ with non-zero mass -- and making a query on this point then gives away the type of distribution. However, since sampling queries will almost always return the very first element, finding such a $i > 1$ amounts to finding a needle in a haystack (without sampling) or to sampling many times (to get a non-trivial element) -- and thus requires many queries. Before formalizing this intuition, we prove the first two items of the above claims:
\paragraph*{Distance to monotonicity} By moving all elements of $S$ at the beginning of the support (points $2,\dots,\abs{S}+1$), the distribution would be monotone; so in particular
\[
  \totalvardist{D_i}{\textsc{Monotone}}\textit{} \leq \half\cdot 2\abs{S}\cdot\frac{\gamma_n}{\abs{S}}=2\gamma_n = \frac{2}{\log n}, \qquad i\in\{1,2\}
\]
\paragraph*{Difference of entropy} By their definition, for any two $D_1$, $D_2$, we have
\begin{align*}
\abs{H(D_1)-H(D_2)} &= \abs{ \sum_{i = 2}^n D_1(i)\log D_1(i) - \sum_{i = 2}^n D_2(i)\log D_2(i)  } 
= \gamma_n \log K_n = \frac{1}{4}.
\end{align*}

\noindent We now turn to the main item, the indistinguishability:\vspace{-0.5\baselineskip}
\paragraph*{Telling $D_1$ and $D_2$ apart}
Assume we have an algorithm \Tester, which can estimate entropy of distributions that are $\bigO{1/\log n}$-close to monotone up to an additive $1/3$ making $q(n)=\littleO{\log n}$ queries; we claim that \Tester cannot be correct with probability $2/3$. As argued before, we can further assume without loss of generality that $\Tester$ makes exactly $2q$ queries, $q$ sampling queries and $q$ evaluation ones; and that for any $\SAMP$ query, it gets ``for free'' the result of an evaluation query on the sample. Finally, and as the sampling queries are by definition non-adaptive, this also allows us to assume that \Tester starts by making its $q$ \SAMP queries.

Let $B_1$ be the event that one of the $q$ first queries results in sampling an element $i > 1$ (that is, $B_1$ is the event that the ``hogging element'' fails its role). Clearly, $B_1$ has same probability no matter with of the two families the unknown distribution belongs to, and
\begin{equation}
  \probaOf{B_1} = 1-(1-\gamma_n)^q = 1-2^{q\log(1-1/\log n)} \leq 1-2^{-2q/\log n} = \bigO{q/\log n} = \littleO{1}
\end{equation}
so with probability $1-\littleO{1}$, $\bar{B_1}$ holds. We further condition on this: i.e., the testing algorithm only saw the first element (which does not convey any information) after the sampling stage.

The situation is now as follows: unless one of its queries hits one of the relevant points in the uniform set $S$ (call this event $B_2$), \newer{the algorithm will see in both case the same thing} -- a sequence of points with probability zero. But by construction, in both cases, the probability over the (uniform) choice of the support $S$ to hit a relevant point with one query is either $\ell_n/(n-1)=\log n/(n-1)$ or $K_n\ell_n/(n-1) = n^{1/4}\log n/(n-1)$; so that the probability of finding such a point in $n$ queries is at most
\begin{equation}
  \probaOf{B_2} \leq 1-\left(1-\frac{K_n\ell_n}{n-1}\right)^q = \bigO{ \frac{q\log n}{n^{3/4}} } = \littleO{1}
\end{equation}
Conditioning on $\bar{B}_1\cup \bar{B}_2$, we get that $\mathcal{T}$ sees exactly the same transcript if the distribution is drawn from $\mathcal{D}_1$ or $\mathcal{D}_2$; so overall, with probability $1-\littleO{1}$ it cannot distinguish between the two cases -- contradicting the assumption.
\end{proof}

\subsubsection{Upper bound: exponential speedup for \cdfsamp oracles}
We now establish the positive result in the case of algorithms given \cdfsamp query access. Note that Batu et al. \cite{BDK+:05} already consider the problem of getting a (multiplicative) estimate of the entropy of $\D$, under the assumption that the distribution is monotone; and describe (both in the evaluation-only and sample-only models) $\poly\!\log(n)$-query algorithms for this task, which work by recursively splitting the domain in a suitable fashion to get a partition into near uniform and negligible intervals.

The main insight here (in addition to the mere fact that we allow ourself a stronger type of access to $\D$) is to use, instead of an \emph{ad hoc} partition of the domain, a specific one tailored for monotone distributions, introduced by Birg\'e \cite{Birge:87} -- and which crucially \emph{does not depend on the distribution itself}.

\begin{definition}[Oblivious decomposition]\label{def:birge:obl:decomp}
  Given a parameter $\eps>0$, the corresponding \emph{oblivious decomposition of $[n]$} is the partition $\mathcal{I}_\eps=(I_1,\dots,I_\ell)$, where $\ell=\clg{\frac{\log( \eps n + 1)}{\eps}}=\bigTheta{\frac{\log n }{\eps}}$ and $\abs{I_{k+1}}=(1+\eps)\abs{I_{k}}$, $1\leq k < \ell$.
\end{definition}

\noindent For a distribution $\D$ and parameter \eps, define $\bar{\D}_\eps$ to be the \emph{flattened distribution} with relation to the oblivious decomposition $\mathcal{I}_\eps$:
\[ \forall k\in [\ell], \forall i\in I_k,\quad \bar{\D}_\eps(i) = \frac{D(I_k)}{\abs{I_k}} \]
We insist that while $\bar{\D}_\eps$ (obviously) depends on $\D$, the partition $\mathcal{I}_\eps$ itself does not; in particular, it can be computed prior to getting any sample or information about $\D$.

\begin{theorem}[\cite{Birge:87}]\label{theorem:birge:obl:decomp}
 If $\D$ is monotone non-increasing, then $\totalvardist{D}{\bar{\D}_\eps} \leq \eps$.
\end{theorem}

\begin{remark} A proof of this theorem, self-contained and phrased in terms of discrete distributions (whereas the original paper by Birg\'e is primarily intended for continuous ones) can be found in \cite{DDSVV:13} -- Theorem 3.1.
\end{remark}

\begin{corollary}\label{coro:birge:decomposition:robust}
  Suppose $\D$ is \eps-close to monotone non-increasing. Then $\totalvardist{D}{\bar{\D}_\eps} \leq 3\eps$; furthermore,  $\bar{\D}_\eps$ is also \eps-close to monotone non-increasing.
\end{corollary}

Finally, we shall also need the following well-known result relating total variation distance and difference of entropies (see e.g. \cite{Zhang:07}, Eq. (4)):
\begin{fact}[Total variation and Entropy]\label{fact:tv:entropy}
Let $\D_1$, $\D_2$ be two distributions on $[n]$ such that $\totalvardist{\D_1}{\D_2} \leq \alpha$, for $\alpha\in[0,1]$. Then $\abs{H(\D_1)-H(\D_2)} \leq \alpha\log(n-1) +h_2(\alpha) \leq \alpha\log\frac{n}{\alpha} + (1-\alpha)\log\frac{1}{1-\alpha}$, where $h_2$ is the binary entropy function\footnote{\newer{That is, $h_2(p)=-p\log p - (1-p)\log(1-p)$ is the entropy of a Bernoulli random variable with parameter $p$.}}.
\end{fact}

\paragraph*{High-level idea} Suppose we use the oblivious decomposition from Definition~\ref{def:birge:obl:decomp}, with small parameter $\alpha$ (to be determined later), to reduce the domain into $\ell=\littleO{n}$ intervals. Then, we can set out to approximate the entropy of the induced \emph{flat} distribution -- that we can efficiently simulate from the \cdfsamp oracles, roughly reducing the complexity parameter from $n$ to $\ell$; it only remains to use the previous approach, slightly adapted, on this flat distribution. Of course, we have to be careful not to incur too much a loss at each step, where we first approximate $H(D)$ by $H(\bar{\D})$, and then specify our cutoff threshold to only consider significant contributions to $H(\bar{\D})$.

\paragraph*{Details} Consider the Birg\'e decomposition of $[n]$ into $\ell=\bigTheta{\log(n\alpha)/\alpha}$ intervals (for $\alpha$ to be defined shortly). Theorem~\ref{theorem:birge:obl:decomp} ensures the corresponding (unknown) flattened distribution $\bar{\D}$ is $\alpha$-close to $\D$; which, by the fact above, implies that
\begin{equation}
  \abs{ H(\bar{\D}) - H(D) } \leq \alpha\left(\log\frac{n}{\alpha} + 2\right)
\end{equation}
Taking $\alpha\eqdef\bigTheta{\Delta/\log n}$, the \newer{right-hand-side} is at most $\Delta/2$; so that it is now sufficient to estimate $H(\bar{\D})$ to $\pm\Delta/2$, where both sampling and evaluation access to $\bar{\D}$ can easily be simulated from the $\CDFEVAL_{\D}$ and $\SAMP_{\D}$ oracles. But although $\bar{\D}$ is a distribution on $[n]$, its ``actual'' support is morally only the $\ell=\tildeTheta{\log^2 n/\Delta}$. Indeed, we may write the entropy of $\bar{\D}$ as
\begin{align*}
  H(\bar{\D}) = \sum_{k=1}^\ell \sum_{x\in I_k} \bar{\D}(x) \log \frac{1}{\bar{\D}(x)} = \sum_{k=1}^\ell \sum_{x\in I_k} \frac{D(I_k)}{\abs{I_k}} \log \frac{\abs{I_k}}{D(I_k)}
             = \sum_{k=1}^\ell \D(I_k) \log \frac{\abs{I_k}}{D(I_k)} = \shortexpect_{k\sim \bar{\D}}\left[ \log \frac{1}{d_k} \right]
\end{align*}
where $d_k=\frac{D(I_k)}{\abs{I_k}}\approx (1+\alpha)^{-k} \D(I_k)$.

\ifnum\fulldetails=1
  As in the previous section, we can then define a cutoff threshold $\tau$ (for $d_k$) and only estimate $\shortexpect_{k\sim \bar{\D}}\left[ \log \frac{1}{d_k}\indic{d_k\geq\tau} \right]$, for this purpose, we need $\ell\cdot\tau\log 1/\tau$ to be at most $\Delta/4$, i.e.
  \[
    \tau\eqdef \bigTheta{ \frac{ \Delta/\ell }{ \log\Delta/\ell } } = \tildeTheta{ \frac{\Delta^2}{\log^2 n}  }
  \]
  and to get with high probability a $\Delta/4$-approximation, it is as before sufficient to make $m=\bigO{ \Delta^2/\log^2(1/\tau) } = \tildeO{\frac{\log^2\frac{\log n}{\Delta}}{\Delta^2}}$ queries.
\else
  As in the previous section, we can then define a cutoff threshold $\tau$ (for $d_k$) and only estimate $\shortexpect_{k\sim \bar{\D}}\big[ \log \frac{1}{d_k}\indic{d_k\geq\tau} \big]$, for this purpose, we need $\ell\cdot\tau\log 1/\tau$ to be at most $\Delta/4$, i.e.
  $\tau = \tildeTheta{ {\Delta^2}/{\log^2 n}  }$
  and to get with high probability a $\Delta/4$-approximation, it is as before sufficient to make $m=\bigO{ \Delta^2/\log^2(1/\tau) } = \tildeO{{\log^2(\frac{\log n}{\Delta})/}{\Delta^2}}$ queries.
\fi

\begin{theorem}\label{theorem:entropy:estimation:additive:monotone}
In the \cdfsamp access model, there exists an algorithm for monotone distributions estimating the entropy up to an additive $\Delta$, with sample complexity $\tildeO{{\log^2\frac{\log n}{\Delta}}/{\Delta^2}}$.
\end{theorem}

\begin{remark}\label{remark:entropy:estimation:additive:monotone:close}
We remark that the above result and algorithm (after some minor changes in the constants) still applies if $\D$ is only guaranteed to be $\bigO{1/\log n}$-\emph{close} to monotone; indeed, as stated in Corollary~\ref{coro:birge:decomposition:robust}, the oblivious decomposition is (crucially) {robust}, and $\bar{\D}$ will still be $\bigO{\eps}$-close to $\D$.
\end{remark}


\subsection{Additive estimation of support size}
We now turn to the task of estimating the effective support size of the distribution: given the promise that $\D$ puts on every element of the domain either no weight or at least some minimum probability mass $1/n > 0$, the goal is to output a good estimate (up to $\pm\eps n$) of the number of elements in the latter situation.

\begin{theorem}\label{theorem:support:estimation:additive:extended:access}
In the \pdfsamp access model, there exists an algorithm \textsc{Estimate-Support} that, on input a threshold $n\in\N^\ast$ and a parameter $\eps > 0$, and given access to a distribution $\D$ (over an arbitrary set) satisfying 
\[ \min_{x\in\supp D} \D(x) \geq \frac{1}{n} \]
estimates the support size $\abs{\supp D}$ up to an additive $\eps n$, with query complexity $\bigO{\frac{1}{\eps^2}}$.
\end{theorem}
\begin{proof}
  Write $k\eqdef \abs{\supp D}$. We describe \textsc{Estimate-Support} which outputs (w.p. at least 2/3) an estimate as required:
  \begin{description}
    \item[If $\eps > \frac{2}{\sqrt{n\ln 3n}}$:]
The algorithm will draw $m=\clg{\frac{4}{\eps^2}}$ samples $x_1,\dots, x_m$ from $\D$, query their probability mass $\D(x_i)$, and output $\hat{k}=\clg{Y}$, where
  \[
    Y \eqdef \frac{1}{m}\sum_{i=1}^m \frac{\indic{\D(x_i) \geq \frac{1}{n} }}{\D(x_i)}
  \]
    \item[If $\eps \leq \frac{2}{\sqrt{n\ln 3n}}$:] in this case, \textsc{Estimate-Support} just draws $m=n\ln 3n=\bigO{\frac{1}{\eps^2}}$ samples $x_1,\dots, x_m$ from $\D$, and returns the number $\hat{k}$ of distincts elements it got (no query access is needed in this case).
  \end{description}

\paragraph*{Analysis} In the first (and interesting) case, let $\phi$ be the function defined over the coset of $\D$ by $\phi(x)=\frac{1}{\D(x)}\cdot\indic{\D(x) \geq \frac{1}{n} }$, so that $\shortexpect_{x\sim\D}[ \phi(x) ] = \sum_{x: \D(x) > \frac{1}{n}} \D(x)\cdot\frac{1}{\D(x)} = \abs{\setOfSuchThat{x}{\D(x) > \frac{1}{n}}}=k$; and as the r.v. $\phi(x_1),\dots, \phi(x_m)$ are i.i.d and taking value in $[0,n]$, an additive Chernoff bound yields
\[
  \probaOf{ \abs{Y - k }  > \frac{\eps n}{2} } \leq 2e^{-\frac{\eps^2 m}{2}} < \frac{1}{3}
\]
Conditioned on this not happening, $k + \frac{\eps}{2}n \leq Y \leq \hat{k}\leq Y + 1 \leq k + \frac{\eps}{2}n + 1 \leq k + \eps n$ (as $\eps > \frac{2}{n}$), and $\hat{k}$ is as stated.

Turning now to the second case, observe first that the promise on $\D$ implies that $1 \leq k \leq n$. It is sufficient to bound the probability that an element of the support is \emph{never} seen during the $m$ draws -- let $F$ denote this event. By a union bound,
\[
  \probaOf{ F } \leq k\cdot\left(1-\frac{1}{n}\right)^m \leq n e^{n\ln (3n)\ln(1-\frac{1}{n})} \leq n e^{-\ln 3n} = \frac{1}{3} 
\]
so w.p. at least $2/3$, every element of the support is drawn, and \textsc{Estimate-Support} returns (exactly) $k$.
\end{proof}


\subsubsection{Lower bound}

In this subsection, we show that the upper bound of Theorem~\ref{theorem:support:estimation:additive:extended:access} is tight.

\begin{theorem}\label{theorem:tolerant:tester:support:size:lb} In the \pdfsamp access model, \eps-additively estimating support size requires query complexity $\bigOmega{\frac{1}{\eps^2}}$.
\end{theorem}
 \begin{proof}
Without loss of generality, suppose $n$ is even, and let $k=\frac{n}{2}$. For any $p\in[0,1]$, consider the following process $\Phi_p$, which yields a random distribution $\D_p$ on $[n]$ (See Fig.\ref{fig:construction:lb:support:size:D}):
\begin{itemize}
  \item draw $k$ i.i.d. random variables $X_1,\dots, X_k\sim\bernoulli{p}$;
  \item for $i\in[k] $, set $\D(i)=\frac{1}{n}(1+X_i)$ and  $\D(n-i)=\frac{1}{n}(1-X_i)$
\end{itemize}
Note that by construction $\D(i)+D(n-i)=\frac{2}{n}$ for all $i\in[k]$.

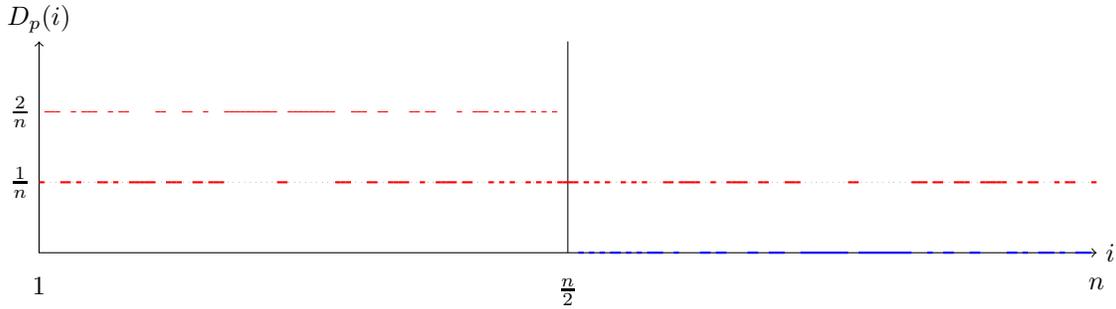
\begin{figure}[!ht]\centering
  \begin{tikzpicture}[x=2pt, y=2pt]
  \pgfmathsetmacro{\xmax}{ 200 }
  \pgfmathsetmacro{\ymax}{40}
  \pgfmathsetmacro{\uniformvalue}{1/3}
    \draw [<->] (0,\ymax) node[above] {$\D_p(i)$} -- (0,0) -- (\xmax,0)  node[right] {$i$};

  \pgfmathsetmacro{\rbuckets}{\xmax} 
  \pgfmathparse{{int(floor(\rbuckets/2)-1))}}\let\rbucketshalf\pgfmathresult
  \pgfmathsetmacro{\buckwidth}{(\xmax/\rbuckets)}
  
      \node [left] at (0,{\ymax*(2*\uniformvalue)}) {$\frac{2}{n}$};
      \node [left] at (0,{\ymax*(\uniformvalue)}) {$\frac{1}{n}$};
      \node [below] at (\xmax/2,-3) {$\frac{n}{2}$};
      \node [below] at (\xmax,-3) {$n$};
      \node [below] at (0,-3) {$1$};
      \draw[ultra thin, black, dotted]  (0,{\ymax*(\uniformvalue)}) -- (\xmax,{\ymax*(\uniformvalue)} );
    \draw[thin] ({\xmax/2},0) -- ({\xmax/2},\ymax);
    
        \foreach \i in {0,...,\rbucketshalf}{
          \pgfmathsetmacro{\cointoss}{floor(random(0,9)/5)}; 
          \pgfmathsetmacro{\fstHigh}{{\ymax*(\uniformvalue*(1+\cointoss))}};
          \pgfmathsetmacro{\sndHigh}{{\ymax*(\uniformvalue*(1-\cointoss))}};
          
          \pgfmathtruncatemacro{\iCoinTossOne}{int(\cointoss))}
          \ifnum\iCoinTossOne=0
            \xdef\linestyle{red};
            \xdef\linethickness{thick};
            \xdef\linethicknesssnd{thick};
          \else
            \xdef\linestyle{blue};
            \xdef\linethickness{solid};
            \xdef\linethicknesssnd{thick};
          \fi
          
          \draw[\linethickness,red]  ({\i*\buckwidth},{\fstHigh}) -- ({(\i+1)*\buckwidth},{\fstHigh});
          \draw[\linethicknesssnd,\linestyle]  ({\xmax-\i*\buckwidth},{\sndHigh}) -- ({\xmax-(\i+1)*\buckwidth},{\sndHigh});
        };
  \end{tikzpicture}\caption{\label{fig:construction:lb:support:size:D}An instance of distribution $\D_p$ with $p=4/10$.}
\end{figure}
Define now, for any $\eps\in(0,1/6)$, the families of distributions $\mathcal{D}^+$ and $\mathcal{D}^-$ induced the above construction, taking $p$ to be respectively $p^+\eqdef \half$ and $p^-\eqdef \half-6\eps$. Hereafter, by $\D^+$ (resp. $\D^-$), we refer to a distribution from $\mathcal{D}^+$ (resp. $\mathcal{D}^-$) generated randomly as above (we assume further, without loss of generality, that $n\gg 1/\eps^2$):
\begin{align*}
 \expect{ \supp{D^+} } &= n - kp^+ = n\left(1-\frac{p^+}{2}\right) = \frac{3}{4}n \\
 \expect{ \supp{D^{\vphantom{+}-}} } &= n - kp^- = n\left(1-\frac{p^-}{2}\right) = \left(\frac{3}{4}+3\eps\right)n
\end{align*}
and, with an additive Chernoff bound,
\begin{align*}
 \probaOf{ \supp{D^+} \geq \frac{3}{4}n + \frac{\eps}{2} n}  &\leq e^{ -\frac{\eps^2 n}{2} } < \frac{1}{100} \\
 \probaOf{ \supp{D^-} \leq \frac{3}{4}n + \frac{5\eps}{2} n} &\leq e^{ -\frac{\eps^2 n}{2} } < \frac{1}{100}
\end{align*}
We hereafter condition on these events $E^+$ and $E^-$ every time we consider a given $\D^+$ or $\D^-$, and set for convenience $s^+\eqdef \frac{3}{4}(n + 2\eps)$, $s^-\eqdef \frac{3}{4}(n + 10\eps)$.

\paragraph*{Reduction}  We shall once again reduce the problem of distinguishing between \textsf{(a)} a fair coin and \textsf{(b)} an $(\half-6\eps)$-biased coin to the problem of approximating the support size: suppose by contradiction we have a tester $\mathcal{T}$ for the latter problem, making $q=\littleO{\frac{1}{\eps^2}}$ queries on input $\eps$.

Given parameter $\eps\in(0,1/100)$ and $\SAMP_\text{coin}$ access to i.i.d. coin tosses coming from one of those two situations (($p^+=\half$, or $p^-=\half-6\eps$), define a distinguisher \Algo as follows:
\begin{itemize}
  \item after picking an even integer $n\gg 1/\eps^2$, $\Algo$ will maintain a set $C\subseteq[n]\times\{0, \frac{1}{n}, \frac{2}{n}\}$ (initially empty), and run $\mathcal{T}$ as a subroutine with parameter $\eps$;
  \item \EVAL: when $\mathcal{T}$ makes an evaluation query on a point $i\in[n]$
  \begin{itemize}
    \item if $i$ has already been committed to (there is a pair $(i,d_i)$ in $C$), it returns $d_i$;
    \item otherwise, it asks for a sample $b$ from $\SAMP_\text{coin}$, and sets
    \[
      d_i = \begin{cases}
        \frac{1}{n} &\text{ if } b = 0\\
        \frac{2}{n} &\text{ if } b = 1 \text{ and } i\in[k]\\
        0 &\text{ if } b = 1 \text{ and } i\in[n]\setminus [k]
      \end{cases}
    \]
    before adding $(i, d_i)$ and $(n-i, \frac{2}{n}-d_i)$ to $C$ and returning $d_i$.
  \end{itemize}
 \item \SAMP: when $\mathcal{T}$ makes an sampling query, \Algo draws u.a.r. $i\sim[k]$, and then proceeds as in the \EVAL case to get $d_i$ and $d_{n-i}$ (that is, if they are not in $C$, it first generates them from a $\SAMP_\text{coin}$ query and commits to them); and then, it returns $i$ w.p. $(n d_i)/2$, and $n-i$ w.p. $(n d_{n-i})/2$.
\end{itemize}
It is easy to see that the process above exactly simulates \pdfsamp access to a distribution $\D$ generated either according to $\Phi_{p^+}$ or $\Phi_{p^-}$ -- in particular, this is true of the sampling queries because each pair $(i, n-i)$ has same total mass $\frac{2}{n}$ under any such distribution, so drawing from $\D$ is equivalent to drawing uniformly $i\in[k]$, and then returning at random $i$ or $n-i$ according to the conditional distribution of $\D$ on $\{i, n-i\}$.

Furthermore, the number of queries to $\SAMP_\text{coin}$ is at most the number of queries made by $\mathcal{T}$ to \Algo, that is $\littleO{\frac{1}{\eps^2}}$. Conditioning on $E^+$ (or $E^-$, depending on whether we are in case \textsf{(a)} or \textsf{(b)}), the distribution $\D$ has support size at most $s^+$ (resp. at least $s^-$). As the estimate $\hat{s}$ that $\mathcal{T}$ will output will, with probability at least $2/3$, be $\eps n$-close to the real support size, and as $s^--s^+ = 2\eps n$, \Algo will distinguish between cases \textsf{(a)} and \textsf{(b)} with probability at least $2/3-2/100 > 6/10$ -- contradicting the fact that $\bigOmega{1/\eps^2}$ samples are required to distinguish between a fair and a $(\half-6\eps)$-biased coin with this probability.
 \end{proof}

\clearpage
\bibliographystyle{plain}	
\bibliography{references}

\appendix
\section{Chernoff Bounds}\label{appendix:chernoff}

\begin{theorem} \label{thm:cb}
Let $Y_1,\dots,Y_m$ be $m$ independent random variables that take on values
in $[0,1]$, where $\expect{Y_i} = p_i$, and
$\sum_{i=1}^m p_i = P$. For any $\gamma \in (0,1]$ we have
\begin{align}
\label{eq:additive-chernoff}
\text{(additive bound)}& &
\probaOf{ \sum_{i=1}^m Y_i   > P+ \gamma m },
\
\probaOf{ \sum_{i=1}^m Y_i  < P - \gamma m }
 &\leq \exp(-2 \gamma^2 m)\\
\label{eq:cher-ub}
\text{(multiplicative bound)}&&
\probaOf{ \sum_{i=1}^m Y_i > (1+\gamma)P } &< \exp(-\gamma^2 P/3)\\
\text{and}\notag\\
\label{eq:cher-lb}
\text{(multiplicative bound)}&&
\probaOf{ \sum_{i=1}^m Y_i < (1-\gamma)P } &< \exp(-\gamma^2 P/2).
\end{align}
\end{theorem}

\noindent The following extension of the multiplicative bound is useful when only upper and/or lower bounds on $P$ are known:
\begin{corollary}\label{cor:cb-upperlower}
In the setting of Theorem~\ref{thm:cb} suppose that
$P_L \leq P \leq P_H.$ Then for any $\gamma \in (0,1]$, we have
\begin{eqnarray}
\probaOf{ \sum_{i=1}^m Y_i > (1+\gamma)P_H } &<& \exp(-\gamma^2 P_H/3)
\label{eq:multCB-upper2}\\
\probaOf{ \sum_{i=1}^m Y_i < (1-\gamma)P_L } &<& \exp(-\gamma^2 P_L/2)
\label{eq:multCB-lower}
\end{eqnarray}
\end{corollary}

\end{document}